%% file: Alpha-conversion.tex
\documentclass[11pt]{memoir}
\usepackage[english]{babel}
\usepackage{amssymb,amsmath,amscd,amsthm,latexsym}
\usepackage[matrix,arrow,curve]{xy}
\CompileMatrices
\usepackage{tikz}
\usetikzlibrary{matrix,arrows}
\usepackage{tikz-cd}
\usepackage{prooftree}
\usepackage{array}
\def\ruleone#1#2{\prooftree #1 \justifies #2 \endprooftree}
\def\ruletwo#1#2#3{\ruleone{#1\quad #2}{#3}}

\theoremstyle{plain}
\newtheorem{theorem}{Theorem}[section]

\newtheorem{lemma}[theorem]{Lemma}

\theoremstyle{definition}
\newtheorem{definition}[theorem]{Definition}

\newtheorem{example}[theorem]{Example}

\numberwithin{equation}{section}
\newcommand{\x}{\mathsf{x}}
\newcommand{\y}{\mathsf{y}}
\newcommand{\z}{\mathsf{z}}
\newcommand{\w}{\mathsf{w}}
\newcommand{\up}{\uparrow\!}
\newcommand{\un}{\underline{1}}

\addto\captionsenglish{}
\begin{document}
\title{Alpha-conversion for lambda terms with explicit weakenings}
\author{George Cherevichenko}
\date{}
\maketitle
\begin{abstract}
\input{abstract}
\end{abstract}
\input{section1}

\input{section2}
\input{section3}

\input{bib}
\end{document}

%% file: abstract.tex
Using explicit weakenings, we can define alpha-conversion by simple equations without any mention of free variables. 

%% file: section1.tex
\section{Terms and alpha-conversion}

\begin{figure}\caption{Terms and functions}\label{Terms}
\begin{framed}
\noindent
\begin{align*}
\x,\y,\z,\w::&= x\mid y\mid z\mid\ldots \tag{Variables}\\
M,N::&= \x \mid   MN \mid \lambda\x M  \mid \,\up M \tag{Terms}\\[20pt]
F::&=   \{\y\x\} \mid F_{\x}  \tag{Functions}
\end{align*}
\begin{align*}
F(MN) &= F(M)F(N)\\
F(\lambda\x M) &=\lambda\x F_{\x}(M)\\
 F_{\x}(\,\up M) &=\, \up F(M)\\
 F_{\x}(\x) &=\x\\
 F_{\x}(\z) &=\,\up F(\z) \tag{$\z\neq\x$}\\
\{\y\x\}(\,\up M) &=\, \up M\\
\{\y\x\}(\x) &= \y\\
\{\y\x\}(\z) &=\, \up\z \tag{$\z\neq\x$}
\end{align*}
\end{framed}
\end{figure}

\begin{figure}\caption{Alpha-conversion}\label{Alpha}
\begin{framed}
\noindent
$\begin{array}{ccc}
M=_{\alpha}M &&\ruleone{M_1=_{\alpha}M_2}{\up M_1=_{\alpha}\up M_2}\\[15pt]
\ruleone{M_1=_{\alpha}M_2}{M_2=_{\alpha}M_1} &&\ruletwo{M_1=_{\alpha}M_2}{N_1=_{\alpha}N_2}{M_1N_1=_{\alpha}M_2N_2}\\[15pt]
\ruletwo{M_1=_{\alpha}M_2}{M_2=_{\alpha}M_3}{M_1=_{\alpha}M_3} &&\ruleone{M_1=_{\alpha}M_2}{\lambda\x M_1=_{\alpha}\lambda\x M_2}\\[15pt]
&\lambda\x M=_{\alpha}\lambda\y\{\y\x\}M
\end{array}$
\end{framed}
\end{figure}

The set of terms $\Lambda\!\up$\, is shown in the Tab.~\ref{Terms}. A term of the form $\up M$ corresponds to usual $M[\up\,]$ (but now we need much less parentheses). For example, $\lambda\x\!\up M$ corresponds to $\lambda\x (M[\up\,])$ and $\up\lambda\x M$ corresponds to $(\lambda\x M)[\up\,]$
We also define a lot of functions $F\colon \Lambda\!\up\,\,\to \Lambda\!\up$\,. Each $F$ has a form $\{\y\x\}$ or $\{\y\x\}_{\z_1\ldots\z_n}$ A function of the form $\{\y\x\}$ roughly corresponds to $[\y/\x]$, but is not the same. It corresponds to $[1\cdot\!\up\,]$ or $\langle Fst,Snd\rangle$ and does nothing except variable renaming. A function of the form $F_{\x}$ corresponds to $[\Uparrow\! F]$. I shall usually write $FM$ instead of $F(M)$. For example,
$\{yx\}\lambda zx$ is shorthand for $\{yx\}(\lambda zx)$

Alpha-conversion is the smallest compatible equivalence relation such that $\lambda\x M=_{\alpha}\lambda\y\{\y\x\}M$ for all $\x,\y,M$ (no restrictions). See Tab.~\ref{Alpha}.

\begin{example}$\lambda xz=_{\alpha}\lambda yz=_{\alpha}\lambda z\!\up z=_{\alpha}\lambda x\!\up z=_{\alpha}\lambda y\!\up z$\\
$\lambda xz=_{\alpha}\lambda z\{zx\}z=\lambda z\!\up z$\\
$\lambda yz=_{\alpha}\lambda z\{zy\}z=\lambda z\!\up z$\\
$\lambda xz=_{\alpha}\lambda x\{xx\}z=\lambda x\!\up z$
\end{example}

\begin{example}$\lambda x\lambda z\,x=_{\alpha}\lambda y\lambda z\,y$\\
$\lambda x\lambda z\,x=_{\alpha}\lambda y\{yx\}\lambda z\,x=\lambda y\lambda z \{yx\}_{z}\,x=\lambda y\lambda z \up\{yx\}\,x=\lambda y\lambda z \up y$\\
$\lambda y\lambda z\,y=_{\alpha}\lambda y\{yy\}\lambda z\,y=\lambda y\lambda z \{yy\}_{z}\,y=\lambda y\lambda z \up\{yy\}\,y=\lambda y\lambda z \up y$

\end{example}

\begin{definition} Context is a finite list of variables, may be with repetitions. $nil$ is the empty context. For example $\Gamma=x,x,y,z$ is a context. $\Gamma,\Delta$ is concatenation of $\Gamma$ and $\Delta$.
\end{definition}

\begin{definition} $\{\x_1\x_2\}_{\Gamma}$ is shorthand for $ \{\x_1\x_2\}_{\y_1\ldots\y_n}$ if $\Gamma=\y_1\ldots\y_n$\\
$\{\x_1\x_2\}_{nil}$ is simply $\{\x_1\x_2\}$ \\
Note that\\
$\{\x_1\x_2\}_{\Gamma,\y}\!\up M=\,\up \{\x_1\x_2\}_{\Gamma} M$\\
$\{\x_1\x_2\}_{\Gamma,\y}\,\y=\y$\\
$\{\x_1\x_2\}_{\Gamma,\y}\,\z=\,\up \{\x_1\x_2\}_{\Gamma}\, \z$ \quad if $\z\neq\y$

\end{definition}

\begin{lemma}\label{Trans}
$\{\x_1\x_2\}\{\x_2\x_3\}M=\{\x_1\x_3\}M$
\end{lemma}

\begin{proof}We shall prove the stronger statement \\
$\{\x_1\x_2\}_{\Gamma}\{\x_2\x_3\}_{\Gamma}\,M=\{\x_1\x_3\}_{\Gamma}\,M$\\
($\Gamma$ is arbitrary) by induction on the structure of $M$. \\[5pt]
Case 1. $M=\lambda\y N$ \\
Prove the induction step\\[5pt]
\ruleone{\{\x_1\x_2\}_{\Gamma,\y}\{\x_2\x_3\}_{\Gamma,\y}\, N=\{\x_1\x_3\}_{\Gamma,\y}\, N}{\{\x_1\x_2\}_{\Gamma}\{\x_2\x_3\}_{\Gamma}\,\lambda\y N=\{\x_1\x_3\}_{\Gamma}\,\lambda\y N}\\[5pt]
This is true because\\
$\{\x_1\x_2\}_{\Gamma}\{\x_2\x_3\}_{\Gamma}\,\lambda\y N=\lambda\y\{\x_1\x_2\}_{\Gamma,\y}\{\x_2\x_3\}_{\Gamma,\y}\, N$\\
$\{\x_1\x_3\}_{\Gamma}\,\lambda\y N=\lambda\y\{\x_1\x_3\}_{\Gamma,\y}\,N$\\[5pt]
Case 2. $M$ is a variable $\z$. Induction on the length of $\Gamma$.\\
Suppose  $\Gamma=nil$ and prove\\
$\{\x_1\x_2\}\{\x_2\x_3\}\z=\{\x_1\x_3\}\z$\\
If $\z=\x_3$ then\\
$\{\x_1\x_2\}\{\x_2\x_3\}\x_3=\x_1=\{\x_1\x_3\}\x_3$\\
If $\z\neq\x_3$ then\\
$\{\x_1\x_2\}\{\x_2\x_3\}\z=\,\up\z=\{\x_1\x_3\}\z$\\
Suppose $\Gamma=\Delta,\y$ and prove the induction step\\[5pt]
\ruleone{\{\x_1\x_2\}_{\Delta}\{\x_2\x_3\}_{\Delta}\,\z=\{\x_1\x_3\}_{\Delta}\,\z}
{\{\x_1\x_2\}_{\Delta,\y}\{\x_2\x_3\}_{\Delta,\y}\,\z=\{\x_1\x_3\}_{\Delta,\y}\,\z}\\[5pt]
If $\z=\y$ then\\
$\{\x_1\x_2\}_{\Delta,\y}\{\x_2\x_3\}_{\Delta,\y}\,\y=\y=\{\x_1\x_3\}_{\Delta,\y}\,\y$\\[5pt]
If $\z\neq\y$ then\\
$\{\x_1\x_2\}_{\Delta,\y}\{\x_2\x_3\}_{\Delta,\y}\,\z=\,\up\{\x_1\x_2\}_{\Delta}\{\x_2\x_3\}_{\Delta}\,\z$\\
$\{\x_1\x_3\}_{\Delta,\y}\,\z=\,\up\{\x_1\x_3\}_{\Delta}\,\z$\\[5pt]
Case 3. $M=\,\up N$\\
If $\Gamma=nil$ then\\
$\{\x_1\x_2\}\{\x_2\x_3\}\up N=\,\up N=\{\x_1\x_3\}\up N$\\
If $\Gamma=\Delta,\y$ then (induction on the structure of $M$)\\[5pt]
\ruleone{\{\x_1\x_2\}_{\Delta}\{\x_2\x_3\}_{\Delta}\, N=\{\x_1\x_3\}_{\Delta}\, N}
{\{\x_1\x_2\}_{\Delta,\y}\{\x_2\x_3\}_{\Delta,\y}\up N=\{\x_1\x_3\}_{\Delta,\y}\up N}\\[5pt]
because\\
$\{\x_1\x_2\}_{\Delta,\y}\{\x_2\x_3\}_{\Delta,\y}\up N=\,\up\{\x_1\x_2\}_{\Delta}\{\x_2\x_3\}_{\Delta}\, N$\\
$\{\x_1\x_3\}_{\Delta,\y}\up N=\,\up\{\x_1\x_3\}_{\Delta}\,N$\\[5pt]
Case 4. $M=N_1N_2$\\
Straightforward, because $F(N_1N_2)=F(N_1)F(N_2)$
\end{proof}

\begin{lemma}\label{Commute} $\{\y\x\} F_{\x} M= F_{\y}\{\y\x\}M$
\end{lemma}

\begin{proof} We shall prove the stronger statement \\
$\{\y\x\}_{\Gamma}\, F_{\x,\Gamma}\, M= F_{\y,\Gamma}\,\{\y\x\}_{\Gamma}\,M$\\
($\Gamma$ is arbitrary) by induction on the structure of $M$.\\[5pt]
Case 1. $M=\lambda\z N$ \\
Prove the induction step\\[5pt]
\ruleone{\{\y\x\}_{\Gamma,\z}\, F_{\x,\Gamma,\z}\, N= F_{\y,\Gamma,\z}\,\{\y\x\}_{\Gamma,\z}\,N}
{\{\y\x\}_{\Gamma}\, F_{\x,\Gamma}\, \lambda\z N= F_{\y,\Gamma}\,\{\y\x\}_{\Gamma}\,\lambda\z N}\\[5pt]
This is true because\\
$\{\y\x\}_{\Gamma}\, F_{\x,\Gamma}\, \lambda\z N=\lambda\z\,\{\y\x\}_{\Gamma,\z}\, F_{\x,\Gamma,\z}\, N$\\
$F_{\y,\Gamma}\,\{\y\x\}_{\Gamma}\,\lambda\z N=\lambda\z\, F_{\y,\Gamma,\z}\,\{\y\x\}_{\Gamma,\z}\,N$\\[5pt]
Case 2. $M$ is a variable $\z$. Induction on the length of $\Delta$.\\
Suppose $\Gamma=nil$ and prove\\
$\{\y\x\} F_{\x}\, \z= F_{\y}\{\y\x\}\z$\\
If $\z=\x$ then\\
$\{\y\x\} F_{\x}\, \x=\y= F_{\y}\{\y\x\}\x$\\
If $\z\neq\x$ then\\
$\{\y\x\} F_{\x}\, \z=\,\up F\z= F_{\y}\{\y\x\}\z$\\
Suppose $\Gamma=\Delta,\w$ and prove the induction step\\[5pt]
\ruleone{\{\y\x\}_{\Delta} F_{\x,\Delta}\, \z= F_{\y,\Delta}\{\y\x\}_{\Delta}\,\z}
{\{\y\x\}_{\Delta,\w} F_{\x,\Delta,\w}\, \z= F_{\y,\Delta,\w}\{\y\x\}_{\Delta,\w}\,\z}\\[5pt]
If $\z=\w$ then\\
$\{\y\x\}_{\Delta,\w} F_{\x,\Delta,\w}\, \w=\w= F_{\y,\Delta,\w}\{\y\x\}_{\Delta,\w}\,\w$\\
If $\z\neq\w$ then\\
$\{\y\x\}_{\Delta,\w} F_{\x,\Delta,\w}\, \z=\,\up\{\y\x\}_{\Delta} F_{\x,\Delta}\, \z$\\
$ F_{\y,\Delta,\w}\{\y\x\}_{\Delta,\w}\,\z=\,\up F_{\y,\Delta}\{\y\x\}_{\Delta}\,\z$\\[5pt]
Case 3. $M=\,\up N$\\
If $\Gamma=nil$ then\\
$\{\y\x\} F_{\x}\! \up N=\,\up FN= F_{\y}\{\y\x\}\!\up N$\\
If $\Gamma=\Delta,\w$ then (induction on the structure of $M$)\\[5pt]
\ruleone{\{\y\x\}_{\Delta} F_{\x,\Delta}\, N= F_{\y,\Delta}\{\y\x\}_{\Delta}\, N}
{\{\y\x\}_{\Delta,\w} F_{\x,\Delta,\w} \up N= F_{\y,\Delta,\w}\{\y\x\}_{\Delta,\w}\up N}\\[5pt]
because\\
$\{\y\x\}_{\Delta,\w} F_{\x,\Delta,\w} \up N=\,\up\{\y\x\}_{\Delta} F_{\x,\Delta}\, N$\\
$F_{\y,\Delta,\w}\{\y\x\}_{\Delta,\w}\up N=\,\up F_{\y,\Delta}\{\y\x\}_{\Delta}\, N$\\[5pt]
Case 4. $M=N_1N_2$\\
Straightforward, because $F(N_1N_2)=F(N_1)F(N_2)$
\end{proof}

\begin{theorem}If $M=_{\alpha}N$ then $F(M)=_{\alpha}F(N)$
\end{theorem}

\begin{proof} We need to prove $F\lambda\x M=_{\alpha}F\lambda\y\{\y\x\} M$\\
$F\lambda\x M=\lambda\x F_{\x}M =_{\alpha}\lambda\y\{\y\x\} F_{\x} M=\lambda\y F_{\y}\{\y\x\} M=F\lambda\y\{\y\x\} M$\\
(the third equality by the previous lemma).
\end{proof}


%% file: section2.tex
\section{de Bruijn's terms}

For each variable $\z$ and term $M$ we define the term $db_{\z}(M)$ such that\linebreak $M=_{\alpha}N$ iff $db_{\z}(M)=db_{\z}(N)$

\begin{definition}$ $\\
$db_{\z}(\x)=\x$\\
$db_{\z}(\up M)=\,\up db_{\z}(M)$\\
$db_{\z}(MN)=db_{\z}(M)db_{\z}(N)$\\
$db_{\z}(\lambda\x M)=\lambda\z\{\z\x\}db_{\z}(M)$
\end{definition}

\begin{example}$ $\\
$db_z(\lambda x\lambda y\,y)=\lambda z\lambda z\, z$ \\
$db_z(\lambda x\lambda y\,x)=\lambda z\lambda z \up z$ \quad($db$ means ``de Bruijn'')\\
$db_z(\lambda y\,x)=\lambda z \up x$
\end{example}

\begin{lemma} $db_{\z}(M)=_{\alpha} M$
\end{lemma}

\begin{proof} Induction on the structure of $M$.\\
If $M=\lambda\x N$ and $db_{\z}(N)=_{\alpha}N$ then\\
$db_{\z}(\lambda\x N)=\lambda\z\{\z\x\}db_{\z}(N)=_{\alpha}\lambda\x\, db_{\z}(N)=_{\alpha}\lambda\x N$

\end{proof}

\begin{theorem} If $db_{\z}(M)=db_{\z}(N)$ then $M=_{\alpha} N$
\end{theorem}

\begin{proof} $M=_{\alpha}db_{\z}(M)=db_{\z}(N)=_{\alpha} N$
\end{proof}

\begin{lemma} $F(db_{\z}M)=db_{\z}(FM)$
\end{lemma}

\begin{proof} Induction on the structure of $M$.\\[5pt]
Case 1. If $M=\lambda\x N$ then\\
$F(db_{\z}(\lambda\x N))=F\lambda\z\{\z\x\}db_{\z}(N)=\lambda\z F_{\z}\{\z\x\}db_{\z}(N)=\lambda\z\{\z\x\} F_{\x}\,db_{\z}(N)$\\
the last equality by Lemma~\ref{Commute}\\
$db_{\z}(F\lambda\x N)=db_{\z}(\lambda\x F_{\x}N)=\lambda\z\{\z\x\}db_{\z}(F_{\x}N)$\\
but $F_{\x}\,db_{\z}(N)=db_{\z}(F_{\x}N)$ by induction hypothesis.\\[5pt]
Case 2. If $M$ is a variable $\y$ then\\
 $F(db_{\z}\y)=F\y$\\
 $db_{\z}(F\y)=F\y$ because $F\y$ always is a variable or $\underbrace{\up\ldots\up}_{n}$ var.\\
 More rigorously, use induction on the structure of $F$ (see Tab.~\ref{Terms}).\\
 Suppose $F$ is $\{\x_1\x_2\}$ and prove\\
 $db_{\z}(\{\x_1\x_2\}\y)=\{\x_1\x_2\}\y$\\
 If $\y=\x_2$ then\\
 $db_{\z}(\{\x_1\x_2\}\x_2)=db_{\z}(\x_1)=\x_1=\{\x_1\x_2\}\x_2$\\
 If $\y\neq\x_2$ then\\
 $db_{\z}(\{\x_1\x_2\}\y)=db_{\z}(\up\y)=\,\up db_{\z}(\y)=\,\up\y=\{\x_1\x_2\}\y$\\
 Now prove the induction step\\[5pt]
 \ruleone{db_{\z}(F\y)=F\y}{db_{\z}(F_{\x}\y)= F_{\x}\y}\\[5pt]
 If $\y=\x$ then \\
 $db_{\z}(F_{\x}\,\x)=db_{\z}(\x)=\x= F_{\x}\,\x$\\
 If $\y\neq\x$ then\\
 $db_{\z}(F_{\x}\,\y)=db_{\z}(\up F\y)=\,\up db_{\z}(F\y)$\\
 $F_{\x}\,\y=\,\up F\y$\\[5pt]
 Case 3. Suppose $M=\,\up N$. \\
 If $F$ is $\{\x_2,\x_1\}$ then \\
 $\{\x_1,\x_2\}\,db_{\z}(\up N)=\{\x_1,\x_2\}\!\up db_{\z}(N)=\,\up db_{\z}(N)$\\
 $db_{\z}(\{\x_1,\x_2\}\!\up N)=db_{\z}(\up N)=\,\up db_{\z}(N)$\\
 If function has the form $F_{\x}$ then (induction on the structure of $M$)\\[5pt]
 \ruleone{db_{\z}(FN)=FN}{db_{\z}(F_{\x}\!\up N)= F_{\x}\!\up N}\\[5pt]
 because\\
 $db_{\z}(F_{\x}\!\up N)=db_{\z}(\up FN)=\,\up db_{\z}(FN)$\\
 $F_{\x}\!\up N=\,\up FN$\\[5pt]
 Case 4. $M=N_1N_2$\\
 Straighforward.

\end{proof}

\begin{theorem} If $M=_{\alpha} N$ then $db_{\z}(M)=db_{\z}(N)$
\end{theorem}

\begin{proof}We need to prove $db_{\z}(\lambda \x M)=db_{\z}(\lambda\y\{\y\x\}M)$\\
$db_{\z}(\lambda \x M)=\lambda\z\{\z\x\}db_{\z}(M)$\\
$db_{\z}(\lambda\y\{\y\x\}M)=$\\
$=\lambda\z\{\z\y\}db_{\z}(\{\y\x\}M)$\\
$=\lambda\z\{\z\y\}\{\y\x\}db_{\z}(M)$\quad (by the previous lemma)\\
$=\lambda\z\{\z\x\}db_{\z}(M)$\qquad\quad (by Lemma~\ref{Trans})
\end{proof}

\begin{theorem} $\lambda\x M=_{\alpha}\lambda\y N$ iff $\{\z\x\}M=_{\alpha}\{\z\y\}N$
\end{theorem}

\begin{proof} Suppose $\{\z\x\}M=_{\alpha}\{\z\y\}N$. Then\\
$\lambda\z\{\z\x\}M=_{\alpha}\lambda\z\{\z\y\}N$\\
 and\\
$\lambda\x M=_{\alpha}\lambda\z\{\z\x\}M=_{\alpha}\lambda\z\{\z\y\}N=_{\alpha}\lambda\y N$\\
Suppose $\lambda\x M=_{\alpha}\lambda\y N$. Then\\
$db_{\z}(\lambda\x M)=db_{\z}(\lambda\y N)$\\
$\lambda\z\{\z\x\}db_{\z}(M)=\lambda\z\{\z\y\}db_{\z}(N)$\\
$\{\z\x\}db_{\z}(M)=\{\z\y\}db_{\z}(N)$\\
and $\{\z\x\}M=_{\alpha}\{\z\y\}N$ because\\
$M=_{\alpha}db_{\z}(M)$\\
$N=_{\alpha}db_{\z}(N)$

\end{proof}


%% file: section3.tex
\section{de Bruijn's terms 2}

\begin{figure}\caption{Inference rules}\label{Inference}
\begin{framed}
\noindent
\begin{align*}
&nil\vdash\x\\[5pt]
&\Gamma,\x\vdash\x\\[5pt]
&\ruleone{\Gamma\vdash \x}{\Gamma,\z\vdash \x} \qquad(\z\neq\x)\\[5pt]
&\ruletwo{\Gamma\vdash M}{\Gamma\vdash N}{\Gamma\vdash MN}\\[5pt]
&\ruleone{nil\vdash M}{nil\vdash\, \uparrow\! M} \\[5pt]
&\ruleone{\Gamma\vdash M}{\Gamma,\x\vdash\, \uparrow\! M} \\[5pt]
&\ruleone{\Gamma,\x\vdash M}{\Gamma\vdash \lambda\x M}
\end{align*}
\end{framed}
\end{figure}

The set of inference rules is shown in Tab.~\ref{Inference}. Note that for each $\Gamma,M$ there is a unique rule with the conclusion $\Gamma\vdash M$

\begin{lemma} For each $\Gamma,M$ there is a unique derivation of $\Gamma\vdash M$
\end{lemma}

\begin{proof} Induction on the structure of $M$.\\[5pt]
Case 1. $M$ is a variable $\x$. Induction over the length of $\Gamma$.\\
If $\Gamma=nil$ then the unique inference is
 $nil\vdash\x$\\
If $\Gamma=\Delta,\z$ then\\
$\Delta,\x\vdash\x$\quad\, if $\z=\x$\\[5pt]
\ruleone{\Delta\vdash \x}{\Delta,\z\vdash \x} \quad\, \text{if} $\z\neq\x$\\[5pt]
Case 2. And so on (nothing more to prove).
\end{proof}

\newpage

\begin{example} The (unique) inference of $nil\vdash \lambda x\lambda y\, xyz$\\[5pt]
\ruleone
{\ruletwo
{\ruletwo
{\ruleone
{x\vdash x}
{x,y\vdash x}
}
{x,y\vdash y}
{x,y\vdash xy}
}
{\ruleone
{nil\vdash z}
{\ruleone
{x\vdash z}
{x,y\vdash z}
}
}
{x,y\vdash xyz}
}
{\ruleone
{x\vdash \lambda y\, xyz}
{nil\vdash \lambda x\lambda y\, xyz}
}

\end{example}

\begin{figure}\caption{Generalized de Bruijn's terms}\label{deBruijn}
\begin{framed}
\noindent
\begin{align*}
\x,\y,\z::&= x\mid y\mid z\mid\ldots \tag{Variables}\\
A,B::&= \x \mid \underline{1} \mid AB \mid \lambda A  \mid \,\up A \tag{Terms}
\end{align*}
\begin{align*}
&\|nil\vdash\x\|=\x\\[5pt]
&\|\Gamma,\x\vdash\x\|=\underline{1}\\[5pt]
&\ruleone{\|\Gamma\vdash \x\|=A}{\|\Gamma,\z\vdash \x\|=\,\uparrow\! A} \tag{$\z\neq\x$}\\[5pt]
&\ruletwo{\|\Gamma\vdash M\|=A}{\|\Gamma\vdash N\|=B}{\|\Gamma\vdash MN\|=AB}\\[5pt]
&\ruleone{\|nil\vdash M\|=A}{\|nil\vdash\, \uparrow\! M\|=\,\uparrow\! A} \\[5pt]
&\ruleone{\|\Gamma\vdash M\|=A}{\|\Gamma,\x\vdash\, \uparrow\! M\|=\,\uparrow\! A} \\[5pt]
&\ruleone{\|\Gamma,\x\vdash M\|=A}{\|\Gamma\vdash \lambda\x M\|=\lambda A}
\end{align*}
\end{framed}
\end{figure}

\begin{definition} The set of generalized de Bruin's terms is shown in Tab.~\ref{deBruijn}. For each $\Gamma,M$ we put in correspondence the generalized de Bruin's term\linebreak $\|\Gamma\vdash M\|$ as shown in Tab.~\ref{deBruijn}. Note that we can write these rules shorter:\\
\begin{align*}
\|nil\vdash\x\| &=\x\\
\|\Gamma,\x\vdash\x\| &=\underline{1}\\
\|\Gamma,\z\vdash \x\| &=\,\up\|\Gamma\vdash \x\|\quad \text{if}\,\,\z\neq\x\\
\|\Gamma\vdash MN\| &=\|\Gamma\vdash M\|\|\Gamma\vdash N\|\\
\|nil\vdash\, \uparrow\! M\| &=\,\up\|nil\vdash M\|\\
\|\Gamma,\x\vdash\, \uparrow\! M\| &=\,\up\|\Gamma\vdash M\|\\
\|\Gamma\vdash \lambda\x M\| &=\lambda\|\Gamma,\x\vdash M\|
\end{align*}
\end{definition}

\begin{example}  $\|nil\vdash \lambda x\lambda y\, xyz\|=\lambda\lambda (\up\underline{1})\underline{1}\up\,\up z$\\
because\\[5pt]
\ruleone
{\ruletwo
{\ruletwo
{\ruleone
{\|x\vdash x\|=\underline{1}}
{\|x,y\vdash x\|=\,\up\underline{1}}
}
{\|x,y\vdash y\|=\underline{1}}
{\|x,y\vdash xy\|=(\up\underline{1})\underline{1}}
}
{\ruleone
{\|nil\vdash z\|=z}
{\ruleone
{\|x\vdash z\|=\,\up z}
{\|x,y\vdash z\|=\,\up\,\up z}
}
}
{\|x,y\vdash xyz\|= (\up\underline{1})\underline{1}\up\,\up z}
}
{\ruleone
{\|x\vdash \lambda y\, xyz\|=\lambda (\up\underline{1})\underline{1}\up\,\up z}
{\|nil\vdash \lambda x\lambda y\, xyz\|=\lambda\lambda (\up\underline{1})\underline{1}\up\,\up z}
}

\end{example}

\begin{theorem} If $M=_{\alpha}N$ then $\|\Gamma\vdash M\|=\|\Gamma\vdash N\|$
for arbitrary $\Gamma$.

\end{theorem}

\begin{proof} We have to prove\\
$\|\Gamma\vdash\lambda\x M\|=\|\Gamma\vdash\lambda\y\{\y\x\}M\|$\\
It is enough to prove\\
$\|\Gamma,\x\vdash M\|=\|\Gamma,\y\vdash\{\y\x\}M\|$\\
because\\
$\|\Gamma\vdash\lambda\x M\|=\lambda\|\Gamma,\x\vdash M\|$\\
$\|\Gamma\vdash\lambda\y\{\y\x\}M\|=\lambda\|\Gamma,\y\vdash\{\y\x\}M\|$\\
We shall prove the stronger statement\\
$\|\Gamma,\x,\Delta\vdash M\|=\|\Gamma,\y,\Delta\vdash\{\y\x\}_{\Delta}\,M\|$\\
for arbitrary $\Delta$. Induction on the structure of $M$.\\[5pt]
Case 1. $M=\lambda\z N$. Prove the induction step\\[5pt]
\ruleone{\|\Gamma,\x,\Delta,\z\vdash N\|=\|\Gamma,\y,\Delta,\z\vdash\{\y\x\}_{\Delta,\z}\,N\|}
{\|\Gamma,\x,\Delta\vdash \lambda\z N\|=\|\Gamma,\y,\Delta\vdash\{\y\x\}_{\Delta}\,\lambda\z N\|}\\[5pt]
This is true because\\
$\|\Gamma,\x,\Delta\vdash \lambda\z N\|=\lambda\|\Gamma,\x,\Delta,\z\vdash N\|$\\
$\|\Gamma,\y,\Delta\vdash\{\y\x\}_{\Delta}\,\lambda\z N\|=\|\Gamma,\y,\Delta\vdash\lambda\z\{\y\x\}_{\Delta,\z}\, N\|=\lambda\|\Gamma,\y,\Delta,\z\vdash\{\y\x\}_{\Delta,\z}\,N\|$\\[5pt]
Case 2. $M$ is a variable $\z$. We have to prove\\
$\|\Gamma,\x,\Delta\vdash \z\|=\|\Gamma,\y,\Delta\vdash\{\y\x\}_{\Delta}\,\z\|$\\
Induction over the length of $\Delta$.\\
Suppose $\Delta=nil$ and prove\\
$\|\Gamma,\x\vdash \z\|=\|\Gamma,\y\vdash\{\y\x\}\z\|$\\
If $\z=\x$ then\\
$\|\Gamma,\x\vdash \x\|=\un$\\
$\|\Gamma,\y\vdash\{\y\x\}\x\|=\|\Gamma,\y\vdash \y\|=\un$\\
If $\z\neq\x$ then\\
$\|\Gamma,\x\vdash \z\|=\,\up \|\Gamma\vdash \z\|$\\
$\|\Gamma,\y\vdash\{\y\x\}\z\|=\|\Gamma,\y\vdash\,\up \z\|=\,\up \|\Gamma\vdash \z\|$\\
Suppose $\Delta=\Sigma,\w$ and prove the induction step\\[5pt]
\ruleone{\|\Gamma,\x,\Sigma\vdash \z\|=\|\Gamma,\y,\Sigma\vdash\{\y\x\}_{\Sigma}\,\z\|}
{\|\Gamma,\x,\Sigma,\w\vdash \z\|=\|\Gamma,\y,\Sigma,\w\vdash\{\y\x\}_{\Sigma,\w}\,\z\|}\\[5pt]
If $\z=\w$ then\\
$\|\Gamma,\x,\Sigma,\w\vdash \w\|=\un$\\
$\|\Gamma,\y,\Sigma,\w\vdash\{\y\x\}_{\Sigma,\w}\,\w\|=\|\Gamma,\x,\Sigma,\w\vdash \w\|=\un$\\
If $\z\neq\w$ then\\
$\|\Gamma,\x,\Sigma,\w\vdash \z\|=\,\up \|\Gamma,\x,\Sigma\vdash \z\|$\\
$\|\Gamma,\y,\Sigma,\w\vdash\{\y\x\}_{\Sigma,\w}\,\z\|=\|\Gamma,\y,\Sigma,\w\vdash\,\up\{\y\x\}_{\Sigma}\,\z\|=
\,\up\|\Gamma,\y,\Sigma\vdash\{\y\x\}_{\Sigma}\,\z\|$\\[5pt]
Case 3. $M=\,\up N$.\\
 Suppose $\Delta=nil$, then \\
$\|\Gamma,\x\vdash\, \up N\|=\|\Gamma,\y\vdash\{\y\x\}\!\up N\|$\\
because\\
$\|\Gamma,\x\vdash\, \up N\|=\,\up \|\Gamma\vdash N\|$\\
$\|\Gamma,\y\vdash\{\y\x\}\!\up N\|=\|\Gamma,\y\vdash\,\up N\|=\,\up \|\Gamma\vdash N\|$\\
Suppose $\Delta=\Sigma,\w$ and prove the induction step (induction on the structure of $M$)\\[5pt]
\ruleone{\|\Gamma,\x,\Sigma\vdash N\|=\|\Gamma,\y,\Sigma\vdash\{\y\x\}_{\Sigma}\,N\|}
{\|\Gamma,\x,\Sigma,\w\vdash\, \up N\|=\|\Gamma,\y,\Sigma,\w\vdash\{\y\x\}_{\Sigma,\x}\!\up N\|}\\[5pt]
This is true because\\
$\|\Gamma,\x,\Sigma,\w\vdash\, \up N\|=\,\up\|\Gamma,\x,\Sigma\vdash N\|$\\
$\|\Gamma,\y,\Sigma,\w\vdash\{\y\x\}_{\Sigma,\x}\!\up N\|=\|\Gamma,\y,\Sigma,\w\vdash\,\up\{\y\x\}_{\Sigma}\, N\|=
\,\up \|\Gamma,\y,\Sigma\vdash\{\y\x\}_{\Sigma}\,N\|$\\[5pt]
Case 4. $M=N_1N_2$\\
Straightforward.

\end{proof}

Now we shall prove that $\|nil\vdash M\|=\|nil\vdash N\|$ implies $M=_{\alpha} N$

\begin{definition} For each $\z$ and $\Gamma$ we define the function $\{\z/\Gamma\}$ as follows\\
$\{\z/nil\}M=M$\\
$\{\z/\x,\Delta\}M=\{\z/\Delta\}\{\z\x\}_{\Delta}M$\\
For example\\
$\{\z/\x\}M=\{\z\x\}M$\\
$\{\z/\x_1,\x_2\}M=\{\z\x_2\}\{\z\x_1\}_{\x_2}M$\\
$\{\z/\x_1,\x_2,\x_3\}=\{\z\x_3\}\{\z\x_2\}_{\x_3}\{\z\x_1\}_{\x_2,\x_3}M$\\
Note that\\
$\{\z/\Delta,\x\}\!\up N=\,\up\{\z/\Delta\}N$\\
$\{\z/\Delta,\x\}\x=\z$\\
$\{\z/\Delta,\x\}\y=\,\up\{\z/\Delta\}\y$\quad if $\y\neq\x$\\
(easy induction)
\end{definition}

\begin{example} $\lambda \x_1\lambda\x_2 M=_{\alpha}\lambda\z\lambda\z\{\z/\x_1,\x_2\}M$\\
$\lambda \x_1\lambda\x_2 M\\
=_{\alpha}\lambda\z\{\z\x_1\}\lambda\x_2 M\\
=\lambda\z\lambda\x_2\{\z\x_1\}_{\x_2}M\\
=_{\alpha}\lambda\z\lambda\z\{\z\x_2\}\{\z\x_1\}_{\x_2}M$\\
\end{example}

\begin{lemma}$\{\z/\Gamma\}\lambda\z\{\z\x\}M=\lambda\z\,\{\z/\Gamma,\x\}M$\label{qq}
\end{lemma}

\begin{proof} Induction over the length of $\Gamma$.\\
Suppose $\Gamma=\y,\Delta$ and prove the induction step\\[5pt]
\ruleone{\{\z/\Delta\}\lambda\z\{\z\x\}N=\lambda\z\{\z/\Delta,\x\}N}
{\{\z/\y,\Delta\}\lambda\z\{\z\x\}M=\lambda\z\,\{\z/\y,\Delta,\x\}M}\\[5pt]
where $N=\{\z\y\}_{\Delta,\x}M$\\[5pt]
$\{\z/\y,\Delta\}\lambda\z\{\z\x\}M=\\
=\{\z/\Delta\}\{\z\y\}_{\Delta}\lambda\z\{\z\x\}M\\
=\{\z/\Delta\}\lambda\z\{\z\y\}_{\Delta,\z}\{\z\x\}M\\
=\{\z/\Delta\}\lambda\z\{\z\x\}\{\z\y\}_{\Delta,\x}M$\quad by Lemma~\ref{Commute}\\
$=\lambda\z\{\z/\Delta,\x\}\{\z\y\}_{\Delta,\x}M$\qquad by induction hypothesis\\
$=\lambda\z\{\z/\y,\Delta,\x\}M$

\end{proof}

\begin{definition} For each generalized de Bruijn's term $A$ and variable $\z$ we put in correspondence the term $db_{\z}(A)$ (note that $db_{\z}(A)\in\Lambda\!\up$\,)
\newpage\par\noindent
$db_{\z}(\x)=\x$\\
$db_{\z}(\underline{1})=\z$\\
$db_{\z}(\lambda A)=\lambda\z\, db_{\z}(A)$\\
$db_{\z}(AB)=db_{\z}(A)db_{\z}(B)$\\
$db_{\z}(\up A)=\,\up db_{\z}(A)$

\end{definition}

\begin{example}$ $\\
$db_{z}(\lambda\lambda (\up\underline{1})\underline{1}\up\,\up z)=\lambda z\lambda z(\up z)z\up\,\up z$\\
$db_{x}(\lambda\lambda (\up\underline{1})\underline{1}\up\,\up z)=\lambda x\lambda x(\up x)x\up\,\up z$
\end{example}

\begin{theorem} $db_{\z}\|nil\vdash M\|=db_{\z}(M)$

\end{theorem}

\begin{proof} We shall prove the stronger statement: for each $\Gamma$\\
$db_{\z}\|\Gamma\vdash M\|=\{\z/\Gamma\}db_{\z}(M)$\\
Induction on the structure of $M$.\\[5pt]
Case 1. $M=\lambda\x N$\\
Prove the induction step\\[5pt]
\ruleone{db_{\z}\|\Gamma,\x\vdash N\|=\{\z/\Gamma,\x\}db_{\z}(N)}
{db_{\z}\|\Gamma\vdash \lambda\x N\|=\{\z/\Gamma\}db_{\z}(\lambda\x N)}\\[5pt]
Note that
$\|\Gamma\vdash \lambda\x N\|=\lambda\|\Gamma,\x\vdash N\|$\\
hence\\
$db_{\z}\|\Gamma\vdash \lambda\x N\|=\\
=\lambda\z\, db_{\z}\|\Gamma,\x\vdash N\|\\
=\lambda\z\, \{\z/\Gamma,\x\}db_{\z}(N)$\qquad by induction hypothesis\\
But\\
$db_{\z}(\lambda\x N)=\lambda\z\{\z\x\}db_{\z}(N)$\\
hence\\
$\{\z/\Gamma\}db_{\z}(\lambda\x N)=\\
=\{\z/\Gamma\}\lambda\z\{\z\x\}db_{\z}(N)\\
=\lambda\z\,\{\z/\Gamma,\x\}db_{\z}(N)$\qquad by Lemma~\ref{qq}\\[5pt]
Case 2. $M$ is a variable $\x$. Prove that\\
$db_{\z}\|\Gamma\vdash \x\|=\{\z/\Gamma\}db_{\z}(\x)$\\
Induction on the length of $\Gamma$.\\
Suppose $\Gamma=nil$ then\\
$db_{\z}\|nil\vdash \x\|=\x=db_{\z}(\x)$\\
Suppose $\Gamma=\Delta,\y$ and prove the induction step\\[5pt]
\ruleone{db_{\z}\|\Delta\vdash \x\|=\{\z/\Delta\}db_{\z}(\x)}
{db_{\z}\|\Delta,\y\vdash \x\|=\{\z/\Delta,\y\}db_{\z}(\x)}\\[5pt]
If $\y=\x$ then\\
$db_{\z}\|\Delta,\x\vdash \x\|=db_{\z}(\un)=\z$\\
$\{\z/\Delta,\x\}db_{\z}(\x)=\{\z/\Delta,\x\}\x=\z$\\[5pt]
If $\y\neq\x$ then\\
$db_{\z}\|\Delta,\y\vdash \x\|=db_{\z}(\up\|\Delta\vdash \x\|)=\,\up db_{\z}(\|\Delta\vdash \x\|)$\\
$\{\z/\Delta,\y\}db_{\z}(\x)=\{\z/\Delta,\y\}\x=\,\up\{\z/\Delta\}\x$\\[5pt]
Case 3. $M=\,\up N$\\
Suppose $\Gamma=nil$ and prove the induction step\\[5pt]
\ruleone{db_{\z}\|nil\vdash N\|=db_{\z}(N)}
{db_{\z}\|nil\vdash\,\up N\|=db_{\z}(\up N)}\\[5pt]
This is true because\\
$db_{\z}\|nil\vdash\,\up N\|=db_{\z}(\up\|nil\vdash N\|)=\,\up db_{\z}\|nil\vdash N\|$\\
$db_{\z}(\up N)=\,\up db_{\z}(N)$\\
Suppose $\Gamma=\Delta,\y$ and prove the induction step (induction on the structure of $M$)\\[5pt]
\ruleone{db_{\z}\|\Delta\vdash N\|=\{\z/\Delta\}db_{\z}(N)}
{db_{\z}\|\Delta,\y\vdash\,\up N\|=\{\z/\Delta,\y\}db_{\z}(\up N)}\\[5pt]
This is true because\\
$db_{\z}\|\Delta,\y\vdash\,\up N\|=db_{\z}(\up\|\Delta\vdash N\|)=\,\up db_{\z}\|\Delta\vdash N\|$\\
$\{\z/\Delta,\y\}db_{\z}(\up N)=\{\z/\Delta,\y\}\!\up db_{\z}(N)=\,\up \{\z/\Delta\}db_{\z}(N)$\\[5pt]
Case 4. $M=N_1N_2$\\
Straightforward.
\end{proof}


%% file: Alpha-conversion.bbl
\begin{thebibliography}{99}
\bibitem{Abadi}
M.Abadi, L.Cardelli, P.L.Curien, J.J.Levy. Explicit substititions. 1991.
\bibitem{Pollack}
R.Pollack. Closure under alpha-conversion. 1993.
\end{thebibliography}
